\documentclass[letterpaper,11pt]{article}
\usepackage[margin=1in]{geometry}
\usepackage{amsmath, amsfonts, amssymb, amsthm}
\usepackage{thm-restate}
\usepackage{bbm,url}
\usepackage[colorlinks=true,urlcolor=black,linkcolor=black,citecolor=black]{hyperref}
\usepackage{graphicx}
\usepackage{color}
\usepackage{subcaption}
\usepackage{hhline}
\usepackage{pifont}

\usepackage{nicefrac,bm}
\usepackage{algorithm}
\usepackage[noend]{algpseudocode}
\usepackage{thmtools}
\usepackage{thm-restate}

\allowdisplaybreaks

\newtheorem{theorem}{Theorem}
\newtheorem*{theorem*}{Theorem}

\newtheorem{lemma}{Lemma}

\theoremstyle{definition}
\newtheorem{definition}{Definition}

\newcommand{\hide}[1]{}

\newcommand{\vect}[1]{\textbf{#1}}
\newcommand{\s}[1]{\mathsf{#1}}

\newcommand{\NW}{\s{NW}}

\newcommand{\x}{\mathbf{x}}
\newcommand{\y}{\mathbf{y}}

\newcommand{\N}{\mathbb N}
\newcommand{\Z}{\mathbb Z}

\newcommand{\R}{\mathbb R}
\newcommand{\poly}[1]{\mathsf{poly}(#1)}

\title{Tractable Fragments of the Maximum Nash Welfare Problem\thanks{Work supported by the NSF Grant CCF-1942321 (CAREER) and ERC Starting Grant ScaleOpt}}

\author{Jugal Garg\footnote{University of Illinois at Urbana-Champaign, USA} \\
\texttt{\small jugal@illinois.edu} 
\and
Edin Husić\footnote{London School of Economics and Political Science, UK}\\
\texttt{\small e.husic@lse.ac.uk}
\and
Aniket Murhekar\footnote{University of Illinois at Urbana-Champaign, USA}\\
\texttt{\small aniket2@illinois.edu}
\and
László Végh\footnote{London School of Economics and Political Science, UK}\\
\texttt{\small l.vegh@lse.ac.uk}
}

\date{}
\begin{document}

\maketitle

\begin{abstract}
We study the problem of maximizing Nash welfare (MNW) while allocating indivisible goods to asymmetric agents. The Nash welfare of an allocation is the weighted geometric mean of agents' utilities, and the allocation with maximum Nash welfare is known to satisfy several desirable fairness and efficiency properties. However, computing such an MNW allocation is NP-hard, even for two agents with identical, additive valuations. Hence, we aim to identify tractable classes that either admit a PTAS, an FPTAS, or an exact polynomial-time algorithm. 

To this end, we design a PTAS for finding an MNW allocation for the case of asymmetric agents with identical, additive valuations, thus generalizing a similar result for symmetric agents \cite{identical-ptas-mnw}. Our techniques can also be adapted to give a PTAS for the problem of computing the optimal $p$-mean welfare. We also show that an MNW allocation can be computed exactly in polynomial time for identical agents with $k$-ary valuations when $k$ is a constant, where every agent has at most $k$ different values for the goods. 
Next, we consider the special case where every agent finds at most two goods valuable, and show that this class admits an efficient algorithm, even for general monotone valuations. In contrast, we note that when agents can value three or more goods, maximizing Nash welfare is NP-hard, even when agents are symmetric and have additive valuations \cite{amanatidis2020mnwefx}, showing our algorithmic result is essentially tight.

Finally, we show that for constantly many asymmetric agents with additive valuations, the MNW problem admits an FPTAS.
\end{abstract}

\section{Introduction}
Fair division of resources is a fundamental problem studied across several disciplines, including computer science and economics. It is concerned with allocating a set of goods among $n$ agents in a fair and efficient manner. Discrete fair division studies the setting where goods are \textit{indivisible}, i.e., cannot be shared among agents. This setting has found several applications, such as allocating donated food items to food banks~\cite{prendergast2016foodbank}, division of inheritance, divorce settlements, course allocation~\cite{othman2010coursealloc}, etc.

Discrete fair division has been the focus of recent work with many remarkable results. We now have a fairly good understanding of the case when agents are \textit{symmetric}, i.e., having equal entitlements or weights. However, many situations demand modeling agents as \textit{asymmetric} and having different entitlements. For instance, fair division of assets of a company among its various stakeholders who hold different stakes. Additionally, an agent can be used to represent a group of individuals with the size of the group being the agent's entitlement. This idea has been used in models for bargaining in committees~\cite{laruelle07mnw-application}, maintaining fairness guarantees among ethnic groups~\cite{benabbou18publichousing}, etc. From a computational viewpoint, fair division of goods to asymmetric agents is relatively less understood, but has gained interest in recent years~\cite{farhadi19entitlements,garg20submodular,chak2020wef1,ChakrabortySS21,Aziz2020survey,BabaioffEF21,GargHV21Rado}.

The concept of Maximum Nash welfare (MNW) turns out to be central in the allocation of goods, both divisible and indivisible. The Nash welfare of an allocation is the weighted geometric mean of the utilities of the agents, and the allocation with maximum Nash welfare is known to satisfy several desirable fairness and efficiency properties~\cite{Kaneko1979TheNS,Moulin2003,caragiannis16nsw-ef1}. 

The Nash welfare arose in several different contexts, including Nash's solution to the bargaining problem~\cite{nashbargaining}, proportional fairness in networking~\cite{kelly97charging}, and Varian's work on competitive equilibrium with equal incomes~\cite{VARIAN1974-efpo}. In the case of divisible goods, an MNW allocation satisfies the strong fairness property of weighted envy-freeness (wEF) and the efficiency property of Pareto-optimality (PO). When goods are indivisible and agents are symmetric, Caragiannis et al.~\cite{caragiannis16nsw-ef1} showed the `unreasonable fairness' of an MNW allocation by proving that it satisfies a relaxation of EF called envy-freeness up to one good (EF1) and is PO. These desirable properties also carry over to the asymmetric case, where the MNW allocation is PO and satisfies a relaxation of wEF called weak-weighted envy-freeness up to one good (wwEF1)~\cite{chak2020wef1}. Furthermore, the Nash welfare function has been used axiomatically in many applications where agents have differing entitlements, e.g., in bargaining with committees \cite{laruelle07mnw-application,thomson09mnw-application,chae10mnw-application} and allocation of resources such as water~\cite{houba13mnw-application,degefu16mnw-application}. The popular fair division website Spliddit\footnote{http://www.spliddit.org/} also uses the MNW rule to fairly allocate goods.  

Since Nash welfare is a well-motivated, practically useful, and provably fair and efficient solution concept for allocating goods, it is important to study its computational complexity. Computing the Nash optimal allocation is NP-hard even for two identical agents (via a reduction from the Partition problem), and is APX-hard~\cite{lee2015nsw-apx,GargHM19} (hard to approximate) for additive valuations. In light of this hardness, the community turned towards developing approximation algorithms for the problem. A series of works~\cite{cole2015nswapprox,cole17nswapprox,Barman18FFEA,McGlaughlinG20} presented constant-factor approximation algorithms for the MNW problem for symmetric agents with additive valuations. Further, for the case of symmetric agents with identical valuations, Nguyen and Rothe~\cite{identical-ptas-mnw} showed that the MNW problem admits a PTAS, which computes an allocation with Nash welfare within $\varepsilon$ factor of the optimal in polynomial-time when $\varepsilon$ is a fixed constant. Additionally, Barman et al.~\cite{barman2018binarynsw} gave a greedy 1.061-approximation algorithm for this case. When the number of agents is constant,  Nguyen et al.~\cite{nguyen2014nsw-np} presented an FPTAS for symmetric agents.

While the symmetric case is well-understood, the asymmetric case is much harder. For example, the best approximation factor for the symmetric case is 1.45 \cite{Barman18FFEA}, while $O(n)$ remains the best known approximation factor for the asymmetric case \cite{garg20submodular}. Developing a constant-factor approximation algorithm is still an open problem; similarly, developing a PTAS for asymmetric agents with identical valuations has been open since the work of~\cite{identical-ptas-mnw}. A natural question to ask is which results on the symmetric case extend to the asymmetric one. However, the ideas in the PTAS for MNW with symmetric agents \cite{identical-ptas-mnw}, which in turn rely on the PTAS for scheduling jobs on identical machines in \cite{ptas-identical1998alon}, do not seem to extend to the asymmetric case.

Our first contribution is to show that for asymmetric agents with identical, additive valuations, the MNW problem admits a PTAS. Further, our approach generalizes in two ways. First, we obtain a PTAS for the problem of maximizing the $p$-mean welfare, which is a generalization of welfare functions, including utilitarian, Nash and egalitarian. Secondly, we obtain an exact polynomial-time algorithm when agents have identical $k$-ary valuations with constant $k$, where agents have at most $k$ different values for the goods. This class was recently shown to admit an efficient algorithm for obtaining EF1+PO allocation of goods to non-identical agents~\cite{GargM21}. A common criticism of algorithms for approximating Nash welfare is that the resulting allocation may not retain the fairness guarantees offered by the MNW allocation. Addressing this criticism, we show a novel way to extend our PTAS to obtain an approximately Nash optimal allocation which simultaneously satisfies wwEF1 for the case of identical, asymmetric agents.

The above results and techniques suggest that restricting the number of values of goods makes the MNW problem tractable, at least for identical agents. We investigate if this intuition continues to hold for non-identical agents. Indeed, it is known from~\cite{darmann2014binary,barman2018binarynsw} that when agents have binary valuations, the MNW allocation is polynomial-time computable. 
We study the setting when each agent finds only a small subset of the goods valuable. We present a strongly polynomial-time algorithm for asymmetric non-identical agents when each agent values at most two goods. In fact, our algorithm applies when all agents have monotone valuations, which is a large class of valuation functions generalizing additive valuations. In contrast, we note that the MNW problem is NP-hard, even when agents are symmetric and each agent finds only three goods valuable with additive valuations \cite{amanatidis2020mnwefx}. This shows our algorithmic result is essentially tight. Finally, we show that the FPTAS of Nguyen et al.~\cite{nguyen2014nsw-np} for the MNW problem with constantly many symmetric agents with additive valuations extends to the asymmetric case.

To sum up, we identify and present algorithms for tractable classes of the maximum Nash welfare problem for asymmetric agents; see Table~\ref{tab:results}. These classes, namely identical valuations and few-valuable instances, are useful from a practical aspect too. It is often common for agents to have identical preferences but different priorities, e.g. for expressing preferences for natural resources like water \cite{degefu16mnw-application,houba13mnw-application} or housing~\cite{benabbou18publichousing}. Likewise, it can be easier for agents to assign distinct `rank' to the goods on a small scale of 1 through $k$ instead of assigning them explicit numerical values. These ranks or scores then be transformed into a $k$-valuable instance $\{1,2\dots,k\}$ to $\mathbb{R}_{\ge 0}$, such as Borda scores or Lexicographic scores (e.g. \cite{darmann2014binary} study the MNW problem with such scoring functions). Finally, the number of agents is a small constant in many practical discrete fair division settings. Our results, therefore, contribute towards completing the picture of tractable classes of the computationally-hard, yet important, maximum Nash welfare problem for the relatively less-studied case of asymmetric agents. 

\begin{table}[h]
\centering
\begin{tabular}{|p{0.3\linewidth}||c|c|}\hline
	\textbf{Instance Type} & \textbf{Symmetric} & \textbf{Asymmetric}  \\\hhline{|=||=|=|}
	Additive, Identical, $n=2$ & NP-hard (Partition problem) & NP-hard (Partition problem) \\\hline
	Additive, Identical & PTAS \cite{identical-ptas-mnw} & PTAS \textbf{(Thm~\ref{thm:ptas})} \\\hline
	Additive, Identical, $k$-ary with constant $k$ & $\poly{n,m}$ \textbf{(Thm~\ref{thm:mnwkary})} & $\poly{n,m}$ \textbf{(Thm~\ref{thm:mnwkary})} \\\hline
	Additive, Constant $n$ & FPTAS \cite{nguyen2014nsw-np} & FPTAS \textbf{(Thm~\ref{thm:fptas})} \\\hline
	Monotone, 2-valuable & $\poly{n,m}$ \textbf{(Thm~\ref{thm:2valuable})} & $\poly{n,m}$ \textbf{(Thm~\ref{thm:2valuable})} \\\hline
	Additive, 3-valuable & NP-hard \cite{amanatidis2020mnwefx} & NP-hard \cite{amanatidis2020mnwefx} \\\hline
\end{tabular}
\caption{Best-known results for the MNW problem}\label{tab:results}
\end{table}

\section{Preliminaries}
\paragraph{Problem instance.}
An instance $(N,M,V,\{\eta_i\}_{i\in N})$ of the maximum Nash welfare problem with asymmetric agents is defined by: (i) the set $N = [n]$ of $n\in\N$ agents, (ii) the set $M = [m]$ of $m\in\N$ indivisible goods, (iii) the set $V = \{v_1,\dots,v_n\}$, where $v_i : 2^M \rightarrow \R_{\ge 0}$ is the valuation or utility function of agent $i\in N$, and (iv) the set $\{\eta_i\}_{i\in N}$, where $\eta_i \in \mathbb{R}_{>0}$ is the \textit{weight} or entitlement of agent $i\in N$. We assume valuations are \textit{normalized}, i.e., $v_i(\emptyset) = $ for all $i\in N$. For ease of notation, we write $v_{ij}$ instead of $v_i(\{j\})$ for an agent $i$ and a good $j$.

\begin{definition}\label{def:instancetype} (Instance Type) \normalfont
An instance $(N,M,V,\{\eta_i\}_{i\in N})$ is said to be:
\begin{enumerate}
\item \textit{Additive}, if $\forall i\in N$ and $\forall S\subseteq M$, $v_i(S) = \sum_{j\in S} v_{ij}$.
\item \textit{Monotone}, if for every $i\in N$, for any $S\subseteq M$ and $j\in M$, $v_i(S\cup \{j\}) \ge v_i(S)$.
\item $k$-ary, if $\forall i\in N$, $v_i$ is additive, and $|\{v_{ij} : j\in M\}| \le k$.
\item $k$-valued, if $\forall i\in N$, $v_i$ is additive, and $\exists U \subseteq \N$ with $|U| \le k$ s.t. $\forall i\in N, j\in M$, $v_{ij} \in U$.
\item $k$-valuable, if 
for every $i\in N$, $v_i$ is monotone and there exists a set $T_i \subseteq M$ with $|T_i| \le k$ such that $v_i(S) = v_i(S\cap T_i)$ for every $S\subseteq M$.
\end{enumerate}
\end{definition}
Note that for identical agents with additive valuations, a $k$-valuable instance is also a $k$-valued instance. Also, note that a $k$-valued instance is $k$-ary, but the containment does not hold in the reverse direction.

\paragraph{Allocation.}
An \textit{allocation} $\x$ of goods to agents is an $n$-partition $(\x_1, \dots, \x_n)$ of the goods, where agent $i$ is allotted $\x_i \subseteq M$, and gets a total value of $v_i(\x_i)$.

\paragraph{Pareto-optimality.}
An allocation $\y$ dominates an allocation $\x$ if $v_i(\y_i) \ge v_i(\x_i)$, for all $i\in N$ and there exists $h$ such that $v_h(\y_h) > v_h(\x_h)$. An allocation is said to be \textit{Pareto optimal} (PO) if no allocation dominates it. 

\paragraph{Maximum Nash Welfare.}
The Nash welfare (NW) of an allocation $\x$ is the weighted geometric mean of the agents' utilities under $\x$, i.e.:
\[\NW(\x) = \big(\prod_{i\in N} v_i(\x_i)^{\eta_i}\big)^{1/\sum_{i\in N}\eta_i}.\] 
An allocation $\x^*$ that maximizes the NW is called a MNW allocation or a Nash optimal allocation. An allocation $\x$ is an $\alpha$-approximation for $\alpha\in (0,1]$ if $\NW(\x) \ge \alpha\cdot\NW(\x^*)$. Note that the optimal allocation is invariant under a uniform scaling of the agent weights. We assume that the optimum Nash welfare of an instance is non-zero for the notion of approximation to be meaningful. We note here that when the optimum Nash welfare is zero, not all allocations satisfy the desirable fairness and efficiency notions and a fair allocation must be carefully selected~\cite{caragiannis16nsw-ef1}.

\paragraph{Weak weighted envy-freeness up to one good (wwEF1).} An allocation $\x$ is said to satisfy wwEF1 if for all agents $i$ and $h$, there exists a good $j\in \x_h$ such that:
\[\frac{v_i(\x_i)}{\eta_i} \ge  \frac{v_i(\x_h)}{\eta_h} - \frac{v_{ij}}{\min\{\eta_i,\eta_h\}}\]
It is known that the MNW allocation satisfies wwEF1 \cite{chak2020wef1}.

\section{PTAS for Identical Agents}\label{sec:ptas}

Our main result of this section is:
\begin{theorem}\label{thm:ptas}
The Nash welfare problem for asymmetric agents with identical, additive valuations admits a polynomial time approximation scheme (PTAS).
\end{theorem}

A natural first step would be to extend the approach in \cite{identical-ptas-mnw} to the asymmetric case. The main idea, which relies on the PTAS for scheduling jobs on identical parallel machines \cite{ptas-identical1998alon}, is to round the values of the goods into a fixed constant number of values and solve the rounded instance exactly using Lenstra's algorithm for integer programs. Similar ideas have been also used for developing PTAS for scheduling problems~\cite{ptas-rounding1986shmoys,ptas-rounding1987shmoys}.

However, this approach does not extend directly for the asymmetric case. Hence, we instead develop a PTAS inspired from the PTAS for scheduling jobs on uniformly related parallel machines~\cite{ptas99epstein}. One can first try to apply their PTAS directly as a black-box for the MNW problem by taking the logarithm of the Nash welfare objective. For additive valuations, this translates to replacing $v_i(\x_i)$ with $\log v_i(\x_i)$. First, it is unclear what to do when $v_i(\x_i) = 0$. Next, their PTAS gives an allocation which \textit{approximates the logarithm} of the optimal Nash welfare to a (multiplicative) factor of $\varepsilon$, which does not translate to a multiplicative $\varepsilon$-approximation for the Nash welfare objective. Hence their PTAS cannot be used directly and must be carefully adapted for the MNW problem. 
We need new definitions (Def 5) and new technical lemmas (Lemmas 1, 6, and proof of Theorem 1).

The idea is to first round the values of all goods into a fixed constant number of values. Such a rounding leaves us with at most a polynomially many types of different subsets, which can be enumerated efficiently. We then construct a layered directed graph where vertices on the $i^{th}$ level represent the union of the allocations of the first 
$i$ agents, and an edge from the $(i-1)^{th}$ to the $i^{th}$ layer represents the contribution of agent $i$ to the Nash welfare. Thus, computing an MNW allocation turns out to be equivalent to searching for a path that maximizes the product of weights, which in turn is equivalent to a shortest path computation in an acyclic graph, which admits efficient algorithms. Finally, we show how to recover an allocation of the original instance from the rounded instance which provides a good approximation to the optimal Nash welfare. 

Since the agents have identical valuations, we use $u_j$ to denote the (common) value of a good $j\in M$ for all the agents. Further, we can assume w.l.o.g. that all $u_j > 0$. Suppose we order the agents so that their weights are non-decreasing. We first show that there is an MNW allocation in which the (non-zero) valuations of the agents are also non-decreasing, thus restricting the search space of allocations. Formally:

\begin{lemma}\label{lem:ordering}
Let the agents be ordered so that the weights $\eta_i$ are non-decreasing with $i$. Then, there exists an MNW allocation $\x$ s.t. for any $1\le i < j \le n$, with $v_i(\x_i),v_j(\x_j)>0$, we have $v_i(\x_i)\le v_j(\x_j)$.
\end{lemma}
\begin{proof}
Suppose for some $1\le i < j \le n$, we have $v_i = v_i(\x_i) > v_j(\x_j) = v_j > 0$ for an MNW allocation $\x$. Consider the allocation $\x'$ obtained from $\x$ by swapping the bundles of $i$ and $j$. Now notice: 
\begin{equation*}
\begin{aligned}
&\log \NW(\x') - \log \NW(\x) = \frac{\log  (v_j^{\eta_i} v_i^{\eta_j}) - \log (v_i^{\eta_i} v_j^{\eta_j})}{\sum_k \eta_k} = \frac{\eta_j - \eta_i}{\sum_k \eta_k}\log(v_i/v_j) \ge 0,
\end{aligned}
\end{equation*}
since $\eta_i \le \eta_j$ and $v_i > v_j$. Thus, the Nash welfare of $\x'$ is at least as much as that of $\x$. Hence, from any optimal $\x$ we can arrive at an optimal allocation with ordered valuations by performing at most $n-1$ such swaps.
\end{proof}

\subsection{Key Definitions}
We now introduce several terms needed for our PTAS. Let $\delta > 0$ be a constant such that $\lambda = 1/\delta$ is an even integer. Here, $\delta$ represents the relative rounding precision.

We approximately represent the goods by rounding their values as follows. We use a parameter $w$, which is either 0 or an integral power of 2, to represent the order of magnitude. Given $w$, we replace each good of value more than $\delta w$ with another good with value rounded to the next higher integral multiple of $\delta^2 w$. The remaining `small' goods are replaced with an appropriate number of goods of value $\delta w$. Thus, a set of goods can be represented using a vector $\bm{m} = (m_\lambda,m_{\lambda+1}\dots,m_{\lambda^2})$ where $m_i$ denotes the number of rounded goods of value $i\delta^2 w$. Such a vector along with $w$ is called a configuration. We use $w=0$ only for representing the empty set. 

Our algorithm proceeds agent by agent, where we use the smallest possible $w$ to represent the set of goods allocated so far (these are the \textit{principal} configurations defined below). As we proceed, the order of magnitude may increase. To avoid rounding the values repeatedly which could increase error, we round the values of a good in advance with respect to the largest $w$ such that the value of the good is at least $\delta w$, as described below:

\begin{definition}\label{def:rounding}(Rounding function) \normalfont The rounding function $r(u) : \Z_{\ge 0} \rightarrow \mathbb{R}_{\ge 0}$ is defined as follows. For $u>0$, let $w$ be the largest integral power of 2 s.t. $u > \delta w$. Let $i$ be the smallest integer s.t. $u \le i\delta^2 w$. Then $r(u) = i\delta^2 w$. Further, $r(0) = 0$.
\end{definition}
Note that whenever $w'$ is an integral power of 2 and $\delta w' < u \le w'$ holds, $r(u)$ is an integral multiple of $\delta^2 w'$. To keep the rounding error due to the small-valued goods as low as possible, we always use the smallest possible $w$ to represent the set of goods allocated. For this, we introduce principal configurations.

\begin{definition}\label{def:configs}(Configurations) \normalfont
\begin{itemize}
\item A \textit{configuration} is a pair $\alpha = (w,(m_\lambda,\dots,m_{\lambda^2}))$, where $w$ is 0 or an integral power of 2, and $\bm{m} = (m_\lambda,\dots,m_{\lambda^2})$ is a vector of non-negative integers.
\item The value of a set of goods $A$ is $V(A) = \sum_{j\in A} u_j$. \\ The rounded value of $A$ is $Vr(A) = \sum_{j\in A} r(u_j)$. \\ The value of a configuration $(w,\bm{m})$ is $V(w,\bm{m}) = \sum_{i=\lambda}^{\lambda^2} m_i \cdot i \delta^2 w$. 
\item For a set of goods $A$, a set of \textit{small goods} w.r.t. $w$ is $A(w) = \{j\in A: u_j \le \delta w\}$.
\item A configuration $(w,\bm{m})$ \textit{represents} $A\subseteq M$ if:
\begin{itemize}
    \item[(i)] for all $j\in A$, $u_j\le w$,
    \item[(ii)] for any $\lambda < i \le \lambda^2$, the number of goods $j\in A$ with $r(u_j)=i\delta^2 w$ equals $m_i$.
    \item[(iii)] $m_{\lambda} \in \{\lfloor Vr(A(w))/(\delta w)\rfloor, \lceil Vr(A(w))/(\delta w)\rceil\}$, which is equivalent to $|Vr(A(w)) - m_{\lambda}\delta w|<\delta w$
\end{itemize}
\item A \textit{principal configuration of $A\subseteq M$} is a configuration $\alpha(A) = (w,\bm{m})$ representing $A$ with the smallest possible $w$. There can only be two such configurations (see (iii) above), and of them, the principal configuration is the one with larger $m_{\lambda}$, i.e., $m_{\lambda} = \lceil Vr(A(w))/(\delta w)\rceil$.
\item A configuration $(w,\bm{m})$ is a \textit{principal configuration} if it is the principal configuration of some $A\subseteq M$, i.e., if there exists an $A\subseteq M$ for which $\alpha(A) = (w,\bm{m})$.
\end{itemize}
\end{definition}

Thus, for every set of goods and large enough $w$, there are exactly one or two configurations representing it, and they can be efficiently computed. Similarly, principal configurations can be efficiently computed. Also, observe from above that $w=0$ is used only for representing an empty set.

A principal configuration can represent several subsets that are equivalent from the point of view of Nash welfare. The key idea is to efficiently enumerate principal configurations instead of enumerating the exponentially many subsets of goods. We show how to bound their number.

\begin{lemma}\label{lem:number-configs}
The set $S$ of principal configurations can be efficiently computed. Further, $|S|\le(m+1)^{\lambda^2}$.
\end{lemma}
\begin{proof}
Observe that a configuration $(w,\bm{m})$ is a principal configuration iff (i) $\bm{m} \le \bm{m'}$ coordinate-wise where $(w,\bm{m'})$ is the principal configuration of the set $M(w) := \{j\in M: u_j \le w\}$, and (ii) $m_i > 0$ for some $i>\lambda^2/2$. The first condition ensures that there are sufficiently many goods of each rounded value, and the second condition ensures that any set of goods represented by $(w,\bm{m})$ has total value strictly more than $(\lambda^2/2)\cdot \delta^2 w = w/2$, ensuring that the value of $w$ chosen is the smallest possible.

There can be at most $(m+1)$ possible values of $w$: 0 and $u_j$'s rounded up to powers of 2. The principal configurations corresponding to a given $w$ can be enumerated by first finding a representation $(w,\bm{m'})$ of $M(w)$ and then enumerating all vectors $0\le \bm{m}\le \bm{m'}$ with $m_i > 0$ for some $i>\lambda^2/2$. For each $i$, there can be at most $(m+1)$ values, recall here that $m$ is the total number of goods. Hence there are at most $(m+1)^{\lambda^2-\lambda+2}$ principal configurations. Moreover, this enumeration also takes $O((m+1)^{\lambda^2})$ time.
\end{proof}

Given configurations $c_{i}$ and $c_{i-1}$ representing goods allocated to agents $\{1,\dots,i\}$ and $\{1,\dots,i-1\}$ respectively, in order to compute a set of goods to be allocated to agent $i$ it becomes necessary to compute the difference of configurations. However $c_i$ and $c_{i-1}$ may have differing orders of magnitude. In order to compare them, we may need to \textit{scale} a configuration. As noted in \cite{ptas99epstein}, no such scaling is required when agents are symmetric.

\begin{definition}\label{def:scaling} (Scaling) \normalfont
The scaled configuration for $(w,\bm{m})$ w.r.t. $w'\ge w$ is a configuration $s(w,\bm{m},w') = (w',\bm{m'})$ where $(w',\bm{m'})$ represents the set $K$ containing exactly $m_i $ goods of value $i\delta^2 w$ for each $i=\lambda,\dots,\lambda^2$; and of the (at most) two such configurations, we choose the one satisfying $|Vr(K(w')) - m'_{\lambda}\delta w'|\le \delta w'/2$, breaking ties arbitrarily.
\end{definition}

We show that scaling a configuration representing a set produces another configuration representing the same set. 
Note that this lemma is why in Definition~\ref{def:configs} we allowed for two configurations to represent a set of goods for a given $w$. 

\begin{lemma}\label{lem:scaled-representation}
Let $(w,\bm{m})$ be a configuration representing $A$, and let $(w',\bm{m'}) = s(w,\bm{m},w')$ for $w'\ge w$. Then $(w',\bm{m'})$ also represents $A$.
\end{lemma}
\begin{proof}
The case of $w=0$ or $w'=w$ is trivial. When $w' > w$, condition (i) in the definition of representation (see Definition~\ref{def:configs}) is satisfied since $u_j \le w$ for each $j\in A$. Condition (ii) is also satisfied by the definition of scaled configuration. For condition (iii), we can assume w.l.o.g. $A=A(w')$, i.e., $A$ only contains small goods with value at most $\delta w'$. Let $K$ be the set of goods from the definition of scaled configuration. Since $(w,\bm{m})$ represents $A$, we have $|Vr(A) - Vr(K)| = |Vr(A(w)) - Vr(K(w))| < \delta w \le \delta w'/2$. From the definition of scaled configuration (Definition~\ref{def:scaling}) and the fact that $K = K(w')$ we get $|Vr(K(w')) - m'_{\lambda}\delta w'| \le \delta w'/2$. Adding these bounds gives $|Vr(A) - m'_{\lambda}\delta w'| < \delta w'$, showing that condition (iii) is also satisfied. Thus $(w',\bm{m'})$ also represents $A$.
\end{proof}

The next lemma is crucial. It enables us to find a set $B$ of goods to be allocated to an agent $i$ if we are given the set $A$ of goods already allocated to agents $\{1,\dots,i-1\}$, and the target configuration $(w',\bm{m'})$ representing the goods allocated to agents $\{1,\dots,i\}$.

\begin{lemma}\label{lem:constructing-alloc}
Let $(w,\bm{m})$ be a configuration representing $A\subseteq M$. Let $w'\ge w$ and $(w',\bm{m''}) = s(w,\bm{m},w')$. Let $(w',\bm{m'})$ be a principal configuration satisfying $\bm{m''} \le \bm{m'}$. Then there exists an efficiently computable set of goods $B$ represented by $(w',\bm{m'})$ such that $B\supseteq A$. 
\end{lemma}
\begin{proof}
We construct $B$ from $A$ as follows. For every $\lambda < i \le \lambda^2$, add $m'_i - m''_i$ goods of rounded value $i\delta^2 w$. Since $(w',\bm{m''})$ is a principal configuration we are guaranteed that such a number of goods exists. Finally we add goods with $u_j \le \delta w'$ one by one until their total value exceeds $(m'_{\lambda}-1)\delta w'$. If small goods are added, their total value is at most $m'_{\lambda}\delta w$. If no small good is added, then since $(w',\bm{m''})$ represents $A$ from Lemma~\ref{lem:scaled-representation}, $Vr(A(w')) < (m''_{\lambda}+1)\delta w' \le (m'_{\lambda}+1)\delta w'$. In either case, $(w',\bm{m'})$ represents $B$.
\end{proof}

Finally, we bound the error due to approximation. Here, $B\setminus A$ is the set of goods allocated to some agent, and $A$ is the set of goods allocated to all agents before.

\begin{lemma}\label{lem:succesive-error}
Let $A\subset B \subseteq M$ be sets of goods and $(w,\bm{m})$ and $(w,\bm{m'})$ configurations representing $A$ and $B$ respectively. Then: 
\[V(B\setminus A)-2\delta w < V(w,\bm{m'}-\bm{m}) < (1+\delta)V(B\setminus A) + 2\delta w
\]
\end{lemma}
\begin{proof}
By definition of the rounding function (Definition~\ref{def:rounding}), we have $u_j \le r(u_j) < (1+\delta)u_j$ for every $j\in M$. Adding this for all $j\in B\setminus A$ we get
\begin{equation}\label{eqn:error1}
V(B\setminus A) \le Vr(B-A) < (1+\delta)V(B\setminus A).    
\end{equation}
By definition of representing configuration, any good not in $B(w)$ contributes the same to both $Vr(B\setminus A)$ and $V(w,\bm{m'}-\bm{m})$. We also have $|Vr(A(w))-m_{\lambda}\delta w| < \delta w$ and $|Vr(B(w))-m'_{\lambda}\delta w| < \delta w$. Thus we get:
\begin{equation}\label{eqn:error2}
|Vr(B\setminus A)-V(w,\bm{m'}-\bm{m})| < 2\delta w.    
\end{equation}
Combining \eqref{eqn:error1} and \eqref{eqn:error2} proves the lemma.
\end{proof}

We use the following structure to efficiently compute near-optimal allocations in our PTAS.
\begin{definition}\label{def:config-graph} (Configuration Graph)
\normalfont The configuration graph $G$ is an edge-weighted directed graph. The vertices of $G$ include a source vertex $(0,\alpha(\emptyset))$, a target vertex  $(n,\alpha(M))$ and vertices $(i, c)$ for $1\le i\le n$ and every principal configuration $c \in S$, where the set $S$ of principal configurations satisfies $|S|
\le (m+1)^{\lambda^2}$ from Lemma~\ref{lem:number-configs}.

For any $1\le i \le n$ and any principal configurations $(w,\bm{m})$ and $(w',\bm{m'})$ with $w'\ge w$, and let $\bm{m''}$ be such that $(w',\bm{m''}) = s(w,\bm{m},w')$. $G$ has the following edges: there is an edge from $(i-1,(w,\bm{m}))$ to $(i,(w',\bm{m'}))$ iff either $(w,\bm{m})=(w',\bm{m'})$, or $\bm{m''}\le \bm{m'}$ and $V(w', \bm{m'}-\bm{m''})\ge w'/3$. The cost of such an edge is $V(w', \bm{m'}-\bm{m''})^{\eta_i}$.
\end{definition}

\begin{definition}\label{def:principalrepresentation} \normalfont
Given an allocation $\x$, a path $\{(i,(w_i,\bm{m}_i))\}_{i=0}^n$ in $G$ represents (resp. is a principal representation of) $\x$ if for each $1\le i\le n$, $(w_i,\bm{m}_i)$ represents (resp. is a principal representation of) $\bigcup_{i'=1}^i \x_{i'}$.
\end{definition}

We now formalize a correspondence between allocations and paths in $G$.

\begin{lemma}\label{lem:alloc-path}
Let $\x$ be an allocation under which agents' values are in non-decreasing order. Then its principal representation $\{(i,(w_i,\bm{m}_i))\}_{i=0}^n$ is a path in $G$.
\end{lemma}
\begin{proof}
Let $\bm{m}'_{i-1}$ be such that $(w_i,\bm{m}'_{i-1}) = s(w_{i-1}, \bm{m}_{i-1}, w_i)$ for $1\le i \le n$. By Lemma~\ref{lem:scaled-representation}, $(w_i, \bm{m}_{i-1}')$ represents the set $\bigcup_{i'=1}^{i-1} \x_{i'}$. If for any $i\in [n]$, $\x_i = \emptyset$, then $(w_{i-1}, \bm{m}_{i-1}) = (w_i, \bm{m}_i)$, and by construction $((i-1,(w_{i-1}, \bm{m}_{i-1})), (i,(w_i,\bm{m}_i)))$ is an edge of $G$. Otherwise, since $(w_i, \bm{m'}_{i-1})$ represents $\bigcup_{i'=1}^{i-1} \x_{i'}$ and $(w_i,\bm{m}_i)$ is the principal representation of $\bigcup_{i'=1}^{i} \x_{i'}$, we have $\bm{m'}_{i-1} \le \bm{m}_i$. Thus, $((i-1,(w_{i-1}, \bm{m}_{i-1})), (i,(w_i,\bm{m}_i)))$ is an edge in $G$ once we show $V(w_i, \bm{m}_i - \bm{m'}_{i-1}) \ge w_i/3$. To see this, by the definition of principal representation (Definition~\ref{def:principalrepresentation}) for $(w_i,\bm{m}_i)$, it follows that $\bigcup_{i'=1}^{i} \x_{i'}$ contains a good $j$ of value $u_j > w_i/2$. Since the allocation has non-decreasing values, $V(\x_i) > w_i/2$. By Lemma~\ref{lem:succesive-error}, we have $V(w_i, \bm{m}_i - \bm{m'}_{i-1}) \ge V(\x_i) - 2\delta w_i \ge w_i/3$ when $\delta \le 1/12$.
\end{proof}

\subsection{The PTAS} 
We are now in a position to describe the PTAS, Algorithm~\ref{alg:ptas}.

\begin{algorithm}[h]
\caption{PTAS for the Nash welfare problem for asymmetric agents with identical valuations}\label{alg:ptas}
\textbf{Input:} Nash welfare instance $(N,M,V,\{\eta_i\}_{i\in N})$\\
\textbf{Output:} An integral allocation $\x$
\begin{algorithmic}[1]
\State Order the agents in non-decreasing order of weights $\eta_i$.
\State Construct the graph $G$ (Definition~\ref{def:config-graph}).
\State Find a path $\{(i,(w_i,\bm{m}_i))\}_{i=0}^n$ in $G$ maximizing the product of costs of edges in the path.
\State From such a path construct the following allocation $\x$. Whenever edge $((i-1,(w, \bm{m})), (i,(w,\bm{m})))$ appears in the path, set $\x_i = \emptyset$. Whenever edge  $((i-1,(w_{i-1}, \bm{m}_{i-1})), (i,(w_i,\bm{m}_i)))$ appears in the path, compute $\x$ by applying Lemma~\ref{lem:constructing-alloc} from the beginning of the path.
\end{algorithmic}
\end{algorithm}

We show that the above scheme closely approximates the Nash welfare.

\begin{lemma}\label{lem:approximation}
Let $\{(i,(w_i,\bm{m}_i))\}_{i=0}^n$ be a path in $G$ representing an allocation $\x$, with $(w_{i-1},\bm{m}_{i-1})=(w_i,\bm{m}_i)$ implying $\x_i = \emptyset$. Let $C$ be the Nash welfare of $\x$, and let $C'$ be the Nash welfare (cost) of the path in $G$. Then $|C-C'|\le 8\delta C'$.
\end{lemma}
\begin{proof}
Let $C_i = V(\x_i)$ and $C'_i = V(w_i, \bm{m}_i - \bm{m'}_{i-1})$, where $\bm{m'}_{i-1}$ is such that $(w_i,\bm{m'}_{i-1}) = s(w_{i-1}, \bm{m}_{i-1}, w_i)$. 
If $(w_{i-1},\bm{m}_{i-1})=(w_i,\bm{m}_i)$, then $\x_i = \emptyset$, and $C_i = C'_i = 0$, implying $C=C'=0$. Otherwise, Lemma~\ref{lem:succesive-error} implies $|V(\x_i) - V(w_i, \bm{m}_i - \bm{m'}_{i-1})| \le \delta V(\x_i)+2\delta w_i \le \delta (V(w_i, \bm{m}_i - \bm{m'}_{i-1})+2\delta w_i + 2 w_i)$. With $\delta \le \frac{1}{12}$, using $V(w_i, \bm{m}_i - \bm{m'}_{i-1}) \ge w_i/3$ gives $|C_i - C'_i|\le 8\delta C'_i$. 

Now we have: 
\begin{equation*}
\begin{aligned}
C &= (\prod_{i} C_i^{\eta_i})^{1/\sum_i \eta_i} \le (\prod_{i} ((1+8\delta)C'_i)^{\eta_i})^{1/\sum_i \eta_i} = (1+8\delta)(\prod_{i} {C'_i}^{\eta_i})^{1/\sum_i \eta_i} = (1+8\delta) C'.
\end{aligned}
\end{equation*}

Similarly, we can obtain $C \ge (1-8\delta) C'$, thus showing $|C-C'|\le 8\delta C'$.
\end{proof}

Finally, we can prove Theorem~\ref{thm:ptas}.
\begin{proof}[Proof of Theorem~\ref{thm:ptas}]
Let $\x^*$ be the Nash optimal allocation and $C^*$ its Nash welfare. From Lemma~\ref{lem:alloc-path} the principal representation of $\x^*$ is a path in $G$, let its cost be ${C^*}'$. From Lemma~\ref{lem:approximation}, we know $|C^*-{C^*}'|\le 8\delta {C^*}'$.

Let $\x$ be the allocation returned by Algorithm~\ref{alg:ptas}, with Nash welfare $C$, and let $C'$ be the cost of the corresponding path in $G$. From Lemma~\ref{lem:approximation}, $|C-C'|\le 8\delta C'$. Since $C'$ is the optimal cost of a path in $G$, $C' \ge {C^*}'$. Putting all of the above together we get:
\begin{equation*}
\begin{aligned}
C^* \le (1+8\delta) {C^*}' \le (1+8\delta) C' \le \frac{1+8\delta}{1-8\delta} C,
\end{aligned}
\end{equation*}
thus showing that $C\ge (1-\varepsilon) C^*$, for any given $\varepsilon \in (0,1)$, by choosing $\delta$ such that $1/\delta$ is the smallest even integer at least $\frac{16-8\varepsilon}{\varepsilon}$.

Finally, Lemma~\ref{lem:number-configs} shows that $G$ can be constructed in time $\poly{m}$. Since $G$ is a layered graph, a longest path in $G$ can be found in linear time in the size of $G$\footnote{A longest path from $s$ to $t$ in $G$ is a shortest path from $s$ to $t$ in $-G$, which is the graph obtained by replacing every edge weight $w\ge 0$ in $G$ with $-w$ in $-G$. Since $G$ is acyclic, $-G$ has no negative-weight cycles. Hence a linear-time shortest path algorithm in $-G$ can be used to compute the longest path from $s$ to $t$ in $G$.
}. Further, by Lemma~\ref{lem:constructing-alloc} an allocation can be efficiently constructed from the corresponding path. 
Thus, the time complexity of Algorithm~\ref{alg:ptas} is $\poly{m^{O(1/\varepsilon^2)}}$.
\end{proof}

\subsection{Extensions}
We present two extensions of the ideas in preceding sections. 

The first observation is that we can obtain a PTAS for the problem of maximizing $p$-mean welfare of symmetric agents with identical valuations. For $p\in \R$, the $p$-mean welfare $M_p(\x)$ of an allocation $\x$ is given by:
\[M_p(\x) = \left( \frac{1}{n} \sum_{i=1}^n v_i(\x_i)^ p\right)^{1/p}\]
The PTAS essentially runs Algorithm~\ref{alg:ptas}, but with the cost of a path equal to the $p$-mean welfare of the associated allocation. It is easy to see that the main argument of Lemma~\ref{lem:approximation} applies here as well. Thus:

\begin{theorem}
The maximum $p$-mean welfare problem for symmetric agents with identical, additive valuations admits a polynomial time approximation scheme (PTAS).
\end{theorem}

Our second observation is for identical agents with $k$-ary valuations where $k$ is a constant. Thus, all values belong to some set of size $k$, say $\{u_1,\dots,u_k\}$. Any subset $S\subseteq M$ of goods can be represented as a vector of non-negative integers $\alpha(S) = (m_1,\dots,m_k)$, where $m_i$ denotes the number of goods in the set having value $v_i$. The number of such vectors is at most $(m+1)^k$, which is $\poly{m}$ since $k$ is a constant. We now construct a layered graph $G$ similar to the graph in Definition~\ref{def:config-graph}. The source vertex is $(0,\alpha(\emptyset))$, the target is $(n,\alpha(M))$, and other vertices are $(i,\alpha(S))$ for every $1\le i < n$ and $S\subseteq M$. There is an edge between $(i-1,\bm{m})$ and $(i,\bm{m'})$ iff $\bm{m'}\ge \bm{m}$, with weight $(V(\bm{m'}-\bm{m}))^{\eta_i}$, where $V(\bm{m'}-\bm{m}) = \sum_{j=1}^k (m'_j-m_j)u_j$. We then find a longest path in $G$ with path cost equal to product of edge weights. This can be done in linear time since $G$ is layered. It is easy to see that a path $\{(i,\bm{m}_i)\}_{i=0}^n$ corresponds to the assignment $\x$ where $\x_i$ is such that $\alpha(\x_i) = \bm{m}_i - \bm{m}_{i-1}$, which can be constructed in polynomial time in a manner similar to Lemma~\ref{lem:constructing-alloc}. Hence we conclude:
\begin{theorem}\label{thm:mnwkary}
The Nash welfare problem for asymmetric agents with identical, $k$-ary valuations admits a polynomial-time algorithm for constant $k$.
\end{theorem}

\subsection{Fairness and Efficiency}
We turn to the problem of obtaining an allocation that is both approximately Nash optimal and fair. We show that:
\begin{theorem}\label{thm:ptasfair}
There exists a $\poly{m^{O(1/\varepsilon^2)}}$-time algorithm which computes an allocation that is $\varepsilon$-approximately Nash-optimal and wwEF1 for asymmetric agents with identical, additive valuations.
\end{theorem}

Our algorithm starts with the allocation returned by our PTAS, Algorithm~\ref{alg:ptas}, and while the allocation is not wwEF1, it transfers a good from an carefully selected wwEF1-envied agent to an wwEF1-envying agent. We show that we can only perform polynomially many transfers, and that each transfer increases the Nash welfare. Thus in polynomially-many subsequent steps we can obtain an $\varepsilon$-approximate MNW allocation that is also wwEF1.

We first prove the following lemma, which shows that transferring a good from an wwEF1-envied agent to an wwEF1-envying agent increases the Nash welfare.

\begin{lemma}\label{lem:nswinc}
Let $\x$ be an allocation in which agent $i$ is not wwEF1 towards agent $h$. Let $\x'$ be the allocation obtained by transferring a good $j\in \x_h$ from $h$ to $i$. Then $\NW(\x') > \NW(\x)$.
\end{lemma}
\begin{proof}
Let $v_i = v(\x_i)$ and $v_h = v(\x_h)$. Since $i$ wwEF1-envies $h$, we have for every $g\in \x_h$:
\begin{equation}\label{eqn:wwef1}
\frac{v_i}{\eta_i} < \frac{v_h}{\eta_h} - \frac{u_g}{\min\{\eta_i,\eta_h\}}
\end{equation}
Further, $v(\x'_i) = v_i + u_j$ and $v(\x'_h) = v_h - u_j$, where $j\in\x_h$ is transferred from $h$ to $i$ resulting in $\x'$. Note that 
\[\bigg(\frac{\NW(\x')}{\NW(\x)}\bigg)^{\sum_k\eta_k} = \bigg(1+\frac{u_j}{v_i}\bigg)^{\eta_i} \bigg(1-\frac{u_j}{v_h}\bigg)^{\eta_h}.\]
We consider two cases:

\begin{enumerate}
\item $\eta_i \ge \eta_h$. By Bernoulli's inequality, 
\[\bigg(1+\frac{u_j}{v_i}\bigg)^{\frac{\eta_i}{\eta_h}} \bigg(1-\frac{u_j}{v_h}\bigg) \ge \bigg(1+\frac{\eta_i}{\eta_h}\cdot\frac{u_j}{v_i}\bigg)\bigg(1-\frac{u_j}{v_h}\bigg).\]

Simple algebra shows that:
\[\bigg(1+\frac{\eta_i}{\eta_h}\cdot\frac{u_j}{v_i}\bigg)\bigg(1-\frac{u_j}{v_h}\bigg) > 1 \iff \frac{v_i}{\eta_i} < \frac{v_h-u_j}{\eta_h}.\]
The latter is true since $i$ is not wwEF1 towards $h$, see \eqref{eqn:wwef1}. Hence $\NW(\x') > \NW(\x)$.

\item $\eta_i < \eta_h$. By Bernoulli's inequality, 
\[\bigg(1+\frac{u_j}{v_i}\bigg)\bigg(1-\frac{u_j}{v_h}\bigg)^{\frac{\eta_h}{\eta_i}} \ge \bigg(1+\frac{u_j}{v_i}\bigg)\bigg(1-\frac{\eta_h}{\eta_i}\cdot\frac{u_j}{v_h}\bigg).\]

Once again, simple algebra shows that:
\[\bigg(1+\frac{u_j}{v_i}\bigg)\bigg(1-\frac{\eta_h}{\eta_i}\cdot\frac{u_j}{v_h}\bigg) > 1 \iff \frac{v_i + u_j}{\eta_i} < \frac{v_h}{\eta_h}.\]
The latter is true since $i$ is not wwEF1 towards $h$, see \eqref{eqn:wwef1}. Hence  $\NW(\x') > \NW(\x)$.
\end{enumerate}
Thus, transferring a good from $h$ to $i$ where $i$ wwEF1-envies $h$ increases the Nash welfare.
\end{proof}

Now consider the following algorithm. Given an initial allocation $\x$, sort and relabel the agents so that $v(\x_1)/\eta_1 \le v(\x_2)/\eta_2 \le \dots \le v(\x_n)/\eta_n$. Let $i$ be the smallest index wwEF1-envious agent, and let $h$ be the smallest index agent wwEF1-envied by $i$. Then we transfer a good $g$ from $h$ to $i$. If an agent $\ell$ now envies $i$, then we transfer $g$ from $i$ to $\ell$, and continue. Let $k$ be the last agent who gets $g$ in the above transfers, i.e, $k$ gets an \textit{`extra good'}. If the overall allocation is wwEF1, we stop; and if not, we repeat the above process. 

First, it is easy to see that no agent wwEF1-envies an agent $k$ just after $k$ receives an extra good $g$. This is true because if $k\ge 2$ then this is by design of the algorithm, and if $k=1$, then no agent will wwEF1-envy $k$ after removing $g$. We use this observation to show the following lemma.
\begin{lemma}
Let $i$ be an agent who receives an extra good. Then subsequently to $i$ receiving the extra good, the size of $i$'s bundle does not decrease.
\end{lemma}
\begin{proof}
Let $i$ be an agent who received an extra good. As argued above, just after $i$ received this good, all agents were wwEF1 towards $i$. Now the bundle of $i$ can reduce in size in some subsequent step only if some agent $h$ wwEF1-envies $i$. This can happen only if $h$ loses a good to some agent agent $k$. Let $\x$ be the allocation just before $h$ loses a good to $k$. Because we transfer to the smallest index wwEF1-envious agent, we must have $v_k/\eta_k \le v_i/\eta_i \le v_h/\eta_h$, where $v_i = v(\x_i)$. We also have:
\begin{enumerate}
\item $\frac{v_k}{\eta_k} < \frac{v_h}{\eta_h} - \frac{u_j}{\min(\eta_h,\eta_k)}$ for all $j\in \x_h$ ($k$ wwEF1-envies $h$), and
\item $\frac{v_k}{\eta_k} \ge \frac{v_i}{\eta_i} - \frac{u_{j'}}{\min(\eta_i,\eta_k)}$ for some $j'\in \x_i$ ($k$ doesn't wwEF1-envy $i$).
\end{enumerate}

\noindent We have the following cases:
\begin{enumerate}
\item $\min(\eta_i,\eta_h) \ge \min(\eta_i,\eta_k)$. In this case:
\begin{equation*}
\begin{aligned}
\frac{v_h-u_j}{\eta_h} &\ge \frac{v_h}{\eta_h} - \frac{u_j}{\min(\eta_h,\eta_k)} > \frac{v_k}{\eta_k} \ge \frac{v_i}{\eta_i} - \frac{u_{j'}}{\min(\eta_i,\eta_k)} \ge \frac{v_i}{\eta_i} - \frac{u_{j'}}{\min(\eta_i,\eta_h)},
\end{aligned}
\end{equation*}
thus showing $h$ is wwEF1 towards $i$ after the transfer.
\item $\min(\eta_i,\eta_h) < \min(\eta_i,\eta_k)$. In this case, first suppose $\eta_h = \min(\eta_h,\eta_k)$, in which case the first inequality of Case 1 is equality, and we still obtain that $h$ does not wwEF1-envy $i$ after the transfer. Otherwise we have $\eta_k = \min(\eta_h,\eta_k)$, i.e., $\eta_k < \eta_h$. But then $\min(\eta_i,\eta_h) < \min(\eta_i,\eta_k)$ cannot happen.
\end{enumerate}
In conclusion, if a good is transferred from $h$ to $k$, then $h$ cannot wwEF1-envy $i$ after the transfer. Thus, the agent $i$ who received an extra good cannot later be wwEF1-envied by an agent $h$, implying that $i$ cannot lose a good. Thus the lemma holds.
\end{proof}

With this lemma, we can see that the algorithm terminates in $O(nm)$ transfers. Indeed, whenever the allocation is not wwEF1, performing transfers from envied to envious agents ends up in some agent $i$ receiving an extra good, after which $i$ does not lose any goods. Thus after $O(nm)$ transfers the allocation must be wwEF1. Moreover, Lemma~\ref{lem:nswinc} shows that the Nash welfare does not decrease during transfers performed by the above algorithm. We can therefore initialize the algorithm with the allocation output by our PTAS, thus obtaining an $\varepsilon$-approximate MNW allocation that is also wwEF1. This shows Theorem~\ref{thm:ptasfair}.

\section{Few valuable goods}\label{sec:2valuable}
We now turn to instances where each agent derives positive utility from only a few goods. Recall from Definition~\ref{def:instancetype} that an instance $(N,M,V,\{\eta_i\}_{i\in N})$ is said to be $t$-valuable if for all agents $i\in N$, there exists a set $T_i \subseteq M$ with $|T_i| \le t$ such that $v_i(S) = v_i(S\cap T)$ for every $S\subseteq M$.

We first show that for 2-valuable instances, where each agent finds at most two goods valuable, a Nash optimal allocation can be computed in strongly polynomial-time, even for asymmetric agents with arbitrary monotone valuations.
\begin{theorem}\label{thm:2valuable}
The Nash welfare problem for 2-valuable instances with asymmetric agents with arbitrary monotone valuations admits a strongly polynomial-time algorithm.
\end{theorem}

Let us first assume that $\NW(\x^*) > 0$, where $\x^*$ is an MNW allocation. If $\NW(\x^*) = 0$, then we will detect it in the following discussion and simply return any allocation. Construct a bipartite graph $G = (N,M,E)$ with agents on one side and goods on the other, with an edge between $i\in N$ and $j\in M$ iff $j\in T_i$. Note that each $i\in N$ has degree at most 2. For any $j\in M$ with only one neighbor $i\in N$, we can straightaway assign $j$ to $i$, since every other agent derives no value from $j$, and doing so does not hurt $i$ due to monotonicity of valuations. Similarly, for any $i\in N$ with only one neighbor $j$, i.e., $T_i = \{j\}$, we assign $j$ to $i$, since this is the only good that can give positive utility to $i$. Let $N'$ be the set of agents which have received goods in this manner. If an agent $i\in N'$ receives two goods $\{j_1,j_2\}$, then we can safely remove $i,j_1,j_2$ from $G$ since she has received both goods she values. Thus, we are left with an updated bipartite graph $G$, an updated set of agents $N$ (and its subset $N'$), and a partial allocation in which each agent $i\in N'$ receives a single good $g_i$, and each agent not in $N'$ has not received any good.

It is easy to observe that since $\NW(\x^*)>0$, there cannot be a set of $r$ agents who value $s<r$ goods together, which would imply that in any allocation there is an agent with no goods. This can be checked in polynomial-time by constructing a Hall's violator set of agents in $G$.

We now identify if there is a $K_{2,2}$-subgraph of $G$. Suppose $\{i_1,i_2,j_1,j_2\}$ is such a subgraph, where $i_1,i_2\in N$, $j_1,j_2\in M$. Since each $i\in N$ has degree at most 2, $T_{i_1} = T_{i_2} = \{j_1,j_2\}$, i.e., both agents can receive positive utility only from $\{j_1,j_2\}$, and $i_1,i_2\notin N'$. Thus $i_1$ and $i_2$ must each receive one of $\{j_1,j_2\}$. Out of the two allocations possible, we can pick the one with higher Nash welfare by comparing $v_{i_1}(j_1)^{\eta_{i_1}}\cdot v_{i_2}(j_2)^{\eta_{i_2}}$ and $v_{i_1}(j_2)^{\eta_{i_1}}\cdot v_{i_2}(j_1)^{\eta_{i_2}}$. Since a $K_{2,2}$ subgraph of $G$ can be identified in $O(n^2)$ time and there can be at most $O(n^2)$ of them, we can identify and eliminate all such $K_{2,2}$-subgraphs of $G$ in polynomial time by selecting the allocation with higher Nash welfare for each such subgraph. 

Suppose we now have an updated $K_{2,2}$-free bipartite graph $G$. Once again, we identify agents $i\notin N'$ who have not received a good so far, and have a single neighbor $j$. Clearly, $j$ must be allocated to $i$ for $i$ to get positive utility. We do so, and add $i$ to $N'$. Finally we remove agents who have received both their valuable goods. This procedure only deletes edges of $G$, and hence does not add new $K_{2,2}$-subgraphs. Therefore, eventually we obtain a $K_{2,2}$-free bipartite graph and a set of agents $N'$ who each have received a single good, and each agent not in $N'$ has not received any good.

From the above discussion, we can infer that any agent $i\notin N'$ must have a degree of 2 in $G$, and any agent $i\in N'$ must have a degree of 1. 
We update $N$ and $M$ to be the sets of remaining agents and goods respectively. We now decide how to allocate the goods in $M$ to agents in $N$ by constructing another weighted graph $H = (N\sqcup M, E', w)$ with vertices $N\cup M$. For each $i\in N'$ with $j$ as its neighbor in $G$, we add an $(i,j)$ edge in $H$ with weight $\eta_i\log(v_i(\{g_i,j\}))$, where $g_i$ is the good $i$ has already received from the previous phases. For each $i\notin N'$ with $j,j'$ as its neighbors in $G$, we add three edges in $H$: edge $(i,j)$ with weight $C + \eta_i\log(v_i(\{j\}))$, edge $(i,j')$ with weight $C + \eta_i\log(v_i(\{j'\}))$, and edge $(j,j')$ with weight $C + \eta_i\log(v_i(\{j,j'\}))$, where $C$ is a large constant with value more than $\sum_i \eta_i \log(v_i(T_i))$. Since $G$ is $K_{2,2}$-free, $H$ is a simple graph, i.e., there is at most one edge between any pair of vertices.

We now compute a maximum weight matching $E'\subseteq E$ in $H$ in polynomial-time. Using this matching, we recover the MNW allocation as follows: For every edge $(i,j)\in E'$, we assign $j$ to $i$. For every edge $(j,j')\in E'$, we assign both $j$ and $j'$ to the unique agent $i$ who values both $j$ and $j'$. 

To see correctness, let $\x$ be any allocation of $M$ to $N$ with positive Nash welfare. Consider the following set of edges $E'$ in $H$: For every $i$ s.t. $\x_i  = \{j,j'\}$, add $(j,j')$ to $E'$. For every $i$ s.t. $\x_i = \{j\}$, add $(i,j)$ to $E'$. We first show that:
\begin{lemma}\label{lem:matching-weight}
$E'$ is a matching in $H$ with weight $|N\setminus N'|C + (\sum_{i}\eta_i)\log\NW(\x)$.
\end{lemma}
\begin{proof}
For sake of contradiction assume $E'$ is not a matching, and two edges of $E'$ share a vertex. If $(i,j), (i,j') \in E$ for $j\neq j'$, then $\x_i = \{j\} = \{j'\}$, which is not possible if $j\neq j'$. Similarly if $(i,j),(i',j)\in E'$, then $\x_i = \x_{i'} = \{j\}$, which is not possible since $\x$ is an allocation. Finally, if $(i,j),(j,j')\in E$, then let $i'$ be the unique agent who values both $j$ and $j'$. Then the construction of $E'$ implies $\x_i = \{j\} \subseteq \{j,j'\} = \x_{i'}$, which is not possible for an allocation $\x$. Hence, $E'$ is a matching in $H$.

To compute the weight of the matching, note that we add an edge in $E'$ for every agent in $N$. For any $i\in N'$, the weight of the edge added is $\eta_i\log(v_i(\x_i))$, and for any $i\notin N'$, the weight of the edge added is $C+\eta_i\log(v_i(\x_i))$. Hence, the weight of the matching is $|N\setminus N'|\cdot C+(\sum_{i}\eta_i)\log \NW(\x)$.
\end{proof}

This shows that every allocation with positive Nash welfare corresponds to a matching in $H$. Let $E^*$ be a maximum weight matching $E^*$ in $H$. We construct the following (possibly partial) allocation $\x$ from $E^*$: For every $(i,j)\in E^*$, set $\x_i = \{j\}$, and for every $(j,j')\in E^*$, set $\x_i = \{j,j'\}$, where $i$ is the unique agent who values both $j$ and $j'$. We show in the subsequent proof that if an item $j$ is unallocated in $\x$, then it can be allocated to an agent without affecting the Nash welfare, resulting in an allocation $\x^*$. We show that:
\begin{lemma}\label{lem:alg-correctness}
$\x^*$ is an MNW allocation.
\end{lemma}
\begin{proof}
It is easy to see that $\x$ is a partition of (possibly a subset of) goods of $M$ since $E^*$ is a matching. Suppose some good $j$ is not allocated to any agent in $\x$. Thus, no edge of $E^*$ is incident on $j$. Clearly $j$ must have a degree at least 2 in $G$, otherwise we would have allocated $j$ to its unique neighbor before constructing $H$. Suppose there is an agent $i\in N$ s.t. edge $(i,j)$ in $G$, and no edge of $E^*$ is incident on $i$. Then $E^*\cup \{(i,j)\}$ is a matching in $G$ with higher weight than $E^*$, is a contradiction. Hence, for every agent $i$ which is adjacent to $j$, there is some edge $(i,j')\in E^*$ incident on $i$, for some $j'\in M$. Since $E^*$ is a matching, $(i,j')$ is the only edge of $E^*$ incident on $j'$. Hence $E' = E^*\cup \{(j,j')\} \setminus \{(i,j')\}$ is also a matching in $G$. By monotonicity of valuations, the weight of $(j,j')$ is no less than the weight of $(i,j')$ and hence we get a (weakly)-better matching $E'$ which is incident on $j$. In this manner, we can get an optimal matching which is incident on all goods. From this matching, we can construct the complete allocation $\x^*$.

Now suppose there is an agent $i\in N$ who gets 0 utility under $\x^*$, i.e., $v_i(\x^*_i) = 0$. Clearly, $i\notin N'$ since agents in $N'$ already have received a good. This implies that $E^* \cap \{(i,j),(i,j'),(j,j')\} = \emptyset$. Observe that in any matching in $H$, at most one of three edges $\{(i',g),(i',g'),(g,g')\}$ can be present, where $i'\notin N'$ and $g,g'$ are the goods that $i'$ values. Hence, the weight of $E^*$ is at most $(|N\setminus N'
|-1)C + (\sum_{i\in N': v_i(\x^*_i)>0} \eta_i\log(v_i(\x^*_i)) )$. Let $\Tilde{\x}$ be an MNW allocation where $\NW(\Tilde{\x}) > 0$. From Lemma~\ref{lem:matching-weight}, there is a matching $E'$ in $H$ with weight $|N\setminus N'|C + (\sum_i \eta_i)\log \NW(\Tilde{\x})$. By the choice of $C$, $E'$ has strictly higher weight than $E^*$, which contradicts the optimality of $E^*$. Thus, every agent gets positive utility under $\x^*$, showing $\NW(\x^*)>0$. 

This implies using Lemma~\ref{lem:matching-weight} that $E^*$ has weight $|N\setminus N'|C + (\sum_i w_i)\log \NW(\x^*)$, which is at least $|N\setminus N'|C + (\sum_i \eta_i)\log \NW(\Tilde{\x})$. This shows $\NW(\x^*) \ge \NW(\Tilde{\x})$, implying that $\x^*$ is Nash optimal too.
\end{proof}

Since the algorithm runs in polynomial-time, we conclude that for 2-valuable instances with arbitrary monotone valuations, an MNW allocation is polynomial-time computable.

In contrast, we note that it is NP-hard to compute a MNW allocation for 3-valuable instances with symmetric, additive agents \cite{amanatidis2020mnwefx}.

\section{FPTAS for Constantly Many Agents}
In this section, we extend the FPTAS of~\cite{nguyen2014nsw-np} for constantly many symmetric agents with additive valuations to the asymmetric case. This result is especially interesting in light of the fact that the MNW problem is NP-hard for even $n=2$ identical, symmetric, additive agents (reduction from the Partition problem).

\begin{theorem}\label{thm:fptas}
The Nash welfare problem for constantly many asymmetric agents with additive valuations admits a fully polynomial-time approximation scheme (FPTAS).
\end{theorem}

The first observation is that we can enumerate all possible utility vectors $\vect{v} = (v_1,\dots,v_n)$ in pseudo-polynomial time since $n$ is constant. Intuitively, this is because in any allocation the utility of agent $i$ is some integer between $0$ and $mv_{max}$, where $v_{max} = \max_{i,j} v_{ij}$. Hence there are at most $(1+mv_{max})^n$ feasible utility vectors. Algorithm~\ref{alg:pseudopoly} computes the set $V_m$ by iteratively computing the sets $V_j$ of feasible utility vectors possible by allocating the first $j$ goods, for $j=1$ to $m$. Here $\vect{e}_i$ denotes a vector of size $n$ with $1$ in the $i^{th}$ coordinate and $0$ elsewhere.

\begin{algorithm}[ht]
\caption{Pseudo-polynomial time algorithm for constantly many asymmetric agents}\label{alg:pseudopoly}
\textbf{Input:} Nash welfare instance $(N,M,V,\{\eta_i\}_{i\in N})$\\
\textbf{Output:} An integral allocation $\x$
\begin{algorithmic}[1]
\State $V_0 \gets \emptyset$
\For{$j=1$ to $m$}
\State $V_j \gets \emptyset$
\For{each $\vect{v}\in V_{j-1}$}
\State $V_j \gets V_j \cup \{\vect{v}+v_{ij}\cdot \vect{e}_i \:|\: i\in [n]\}$
\EndFor
\EndFor
\State \Return $\vect{v}\in V_m$ s.t. $(\prod_{i\in[n]} \vect{v}_i^{\eta_i})^{1/\sum_i\eta_i}$
\end{algorithmic}
\end{algorithm}

We now turn this pseudo-polynomial time algorithm into an FPTAS by discarding unnecessary utility vectors which are almost equivalent from the point of view of Nash welfare. Naturally, doing so introduces some error, but we show that the resulting welfare is not too far from optimal. 

Let $\varepsilon \in (0,1)$ be the desired approximation constant. Let $\alpha = 1+\frac{\varepsilon}{2m}$, and let $K = \lceil\log_{\alpha} v_{max}\rceil$. For each $k\in [K]$, define the interval $L_k = [\alpha^{k-1},\alpha^{k}]$. We now define an equivalence relation $\sim$ given by $\vect{x} \sim \vect{y}$ iff for all $1\le i\le n$, $x_i = y_i = 0$ or $x_i,y_i \in L_k$ for some $k\in [K]$; it is easy to see that $\sim$ is an equivalence relation.

Intuitively, we consider all utility vectors in an equivalence class as equivalent from the point of view of Nash welfare. Hence, we modify Algorithm~\ref{alg:pseudopoly} by maintaining reduced sets $V_{j}^*$ instead of $V_j$. We do this by creating $V_j$ from $V_{j-1}^*$ instead of $V_j$ (Lines 4-5 of Algorithm~\ref{alg:fptas}), and then reducing $V_j$ to $V_j^*$ by keeping only one vector from each equivalence class of $\sim$ and discarding the rest (performed by the $\textsc{Reduce}$ subroutine in Line 6 below).

\begin{algorithm}[ht]
\caption{FPTAS for the Nash welfare problem with constantly many asymmetric agents}\label{alg:fptas}
\textbf{Input:} Nash welfare instance $(N,M,V,\{\eta_i\}_{i\in N})$\\
\textbf{Output:} An integral allocation $\x$
\begin{algorithmic}[1]
\State $V^*_0 \gets \emptyset$
\For{$j=1$ to $m$}
\State $V_j \gets \emptyset$
\For{each $\vect{v}^*\in V^*_{j-1}$}
\State $V_j \gets V_j \cup \{\vect{v}^*+v_{ij}\cdot \vect{e}_i \:|\: i\in [n]\}$
\EndFor
\State $V_j^* \gets \textsc{Reduce}(V_j)$
\EndFor
\State \Return $\vect{v}^*\in V^*_m$ s.t. $(\prod_{i\in[n]} {\vect{v}^*_i}^{\eta_i})^{1/\sum_i\eta_i}$
\end{algorithmic}
\end{algorithm}

We show that Algorithm~\ref{alg:fptas} closely approximates the optimal Nash welfare. For this, first notice that for any two vectors $\x \sim \y$, $x_i \ge y_i/\alpha$ holds for every $i\in[n]$. This is clearly true if both $x_i = y_i = 0$. Otherwise, $x_i,y_i\in L_k = [\alpha^{k-1},\alpha^k]$ for some $k\in[K]$, and hence $x_i/y_i \ge \alpha^{k-1}/\alpha^k = 1/\alpha$.

We relate the vectors in sets $V_j$ of Algorithm~\ref{alg:pseudopoly} and sets $V^*_j$ of Algorithm~\ref{alg:fptas}. Our claim is that for every $\vect{v}\in V_j$, there exists a $\vect{v}^*\in V^*_j$ such that for all $i\in[n]$, $v^*_i \ge \frac{1}{\alpha^j}v_i$. We prove this by induction. The claim is clearly true for $j=1$ since $V_1 = V_1^*$. Assume the claim holds for some $j-1$. Consider any $\vect{v}\in V_j$ created in Line 5 of Algorithm~\ref{alg:pseudopoly} by adding some vector $\vect{w}\in V_{j-1}$ and $v_{\ell j}\cdot \vect{e}_\ell$ for some $\ell\in[n]$. From the induction hypothesis, there exists some $\vect{w}^*\in V_{j-1}^*$ s.t. $w^*_i \ge \frac{1}{\alpha^{j-1}}w_j$. Note that Line 5 of Algorithm~\ref{alg:fptas} will add the vector $\vect{w}^* + v_{\ell j} \vect{e}_\ell$ to $V_j$, but it may get removed during the $\textsc{Reduce}$ operation. In any case, there is some vector $\vect{v}^*\in V_{j}^*$ s.t. $\vect{v}\sim \vect{w}^* + v_{\ell j} \vect{e}_\ell$. Thus:
\[v^*_\ell \ge \frac{1}{\alpha}\cdot(w^*_\ell + v_{\ell j}) \ge \frac{1}{\alpha}\cdot\bigg(\frac{1}{\alpha^{j-1}}w_\ell + v_{\ell j}\bigg) \ge \frac{1}{\alpha^j}\cdot v_\ell,\]
and for every $i\neq \ell$:
\[v^*_i \ge \frac{1}{\alpha}\cdot w^*_i \ge \frac{1}{\alpha^{j}}\cdot w_i = \frac{1}{\alpha^{j}}\cdot v_i,\]
thus completing the induction step.

Let $\vect{v}$ be an optimal vector in $V_m$ returned by Algorithm~\ref{alg:pseudopoly} with optimal Nash welfare $OPT$. Correspondingly, there must be another vector $\vect{v}\in V^*_m$ such that for all $i\in[n]$, $v_i^* \ge \frac{1}{\alpha^m} v_i$. Thus:

\[\bigg(\prod_{i\in[n]} {v_i^*}^{\eta_i}\bigg)^{1/\sum_{i}\eta_i} \ge \bigg(\prod_{i\in[n]} {\bigg(\frac{1}{\alpha^m}\cdot v_i}\bigg)^{\eta_i}\bigg)^{1/\sum_{i}\eta_i}\ge \frac{1}{\alpha^m}\bigg(\prod_{i\in[n]} {v_i}^{\eta_i}\bigg)^{1/\sum_{i}\eta_i} = \frac{1}{\alpha^m}\cdot OPT\]

Finally, the desired approximation follows by noting that $\alpha^m = (1+\frac{\varepsilon}{2m})^m \le e^{\varepsilon/2} \le 1+\varepsilon$. Thus:
\[\bigg(\prod_{i\in[n]} {v_i^*}^{\eta_i}\bigg)^{1/\sum_{i}\eta_i} \ge \frac{1}{1+\varepsilon}\cdot OPT \ge (1-\varepsilon)OPT.\]

We now compute the run-time of Algorithm~\ref{alg:fptas} using the fact that $n$ is constant. Notice that each $V^*_j$ is bounded by $(K+1)^n \in O(K^n)$, since this is the maximum number of equivalence classes of $\sim$. Next observe that each $\textsc{Reduce}$ operation takes time $|V_j|^2 \le (n|V^*_j|)^2 \in O(K^{2n})$. Since there are $m$ iterations, the overall run-time is $O(mK^{2n})$. Now observe that 
\[K = \lceil \log_{\alpha} v_{max} \rceil \le \bigg\lceil \frac{\log v_{max}}{\log (1+\frac{\varepsilon}{2m})} \bigg\rceil \le \bigg\lceil \bigg(1+\frac{2m}{\varepsilon}\bigg)\log v_{max} \bigg\rceil,\]
thus showing that $O(mK^{2n})$ is polynomial in the input size $O(n\cdot m\cdot \log v_{max})$ for any fixed constant $\varepsilon$. Further the dependence on $\varepsilon$ is polynomial, thus showing Algorithm~\ref{alg:fptas} is an FPTAS.

\section{Discussion}\label{app:discussion}
In this paper, we studied the problem of maximizing Nash welfare (MNW) while allocating indivisible goods to asymmetric agents. We presented a PTAS when agents have identical valuations, thereby generalizing the existing PTAS for symmetric agents with identical valuations. Our PTAS also extends to the problem of maximizing the $p$-mean welfare, as well as for computing an MNW allocation when agents have $k$-ary valuations for constant $k$. Next, we studied the MNW problem when each agent has a monotone valuation function and finds at most two goods valuable. In this case, we gave a strongly polynomial-time algorithm. Further, when agents are symmetric, have additive valuations and can value 3 or more goods, the MNW problem is NP-hard, essentially showing the tightness of our polynomial-time algorithm. Finally, we showed that when there are constantly many asymmetric agents, the MNW problem admits an FPTAS, thereby generalizing a similar result for the symmetric case. 

Our work, therefore, identifies and presents algorithms for tractable cases of the maximum Nash welfare problem for the less-understood setting of asymmetric agents. We conclude with three interesting open problems. The first is developing a constant-factor approximation algorithm for asymmetric agents with general non-identical valuations. The second is to investigate the complexity of the problem when agents are of constantly many distinct types in terms of valuations functions. When there is only one type the instance is identical and we have a PTAS, does the problem become hard when there are two types of agents? Lastly, it is also open whether there is a polynomial-time algorithm for asymmetric agents with binary valuations, like it is the case when agents are symmetric.

\bibliography{references}
\end{document}